\documentclass[letterpaper]{scrartcl}
\usepackage[total={16.2cm, 23.5cm}]{geometry}
\usepackage{microtype}
\usepackage{amsmath,amsthm,amssymb}
\usepackage{tikz}
\usetikzlibrary{patterns,decorations.pathreplacing, arrows.meta}
\usepackage{hyperref}
\hypersetup{
pdfencoding=auto, psdextra,
colorlinks=true,citecolor=green!50!black,linkcolor=red!60!black,
}
\usepackage[sort&compress,nameinlink]{cleveref}
\usepackage{nicefrac}
\usepackage{natbib}
\usepackage{comment}
\usepackage{paralist}
\usepackage{xcolor}
\usepackage{colortbl}
\usepackage{subcaption}
\usepackage{booktabs}
\usepackage{listings}
\usepackage[linesnumbered,ruled,vlined]{algorithm2e}
\usepackage[linecolor=green!70!black, backgroundcolor=green!10, bordercolor=black, textsize=small]{todonotes}
\definecolor{winered}{rgb}{0.6,0.1,0.1}

\renewcommand*{\leq}{\leqslant}

\renewcommand*{\geq}{\geqslant}
\renewcommand{\epsilon}{\varepsilon}

\crefname{table}{Table}{Tables}
\Crefname{table}{Table}{Tables}
\crefname{figure}{Figure}{Figures}
\crefname{theorem}{Theorem}{Theorems}
\crefname{definition}{Definition}{Definitions}
\crefname{corollary}{Corollary}{Corollaries}
\crefname{observation}{Observation}{Observations}
\crefname{question}{Question}{Question}
\crefname{lemma}{Lemma}{Lemmas}
\crefname{example}{Example}{Examples}
\crefname{reduction}{Reduction}{Reductions}
\crefname{construction}{Construction}{Constructions}
\crefname{subsection}{Section}{Sections}
\crefname{section}{Section}{Sections}
\crefname{proposition}{Proposition}{Propositions}
\crefname{algorithm}{Algorithm}{Algorithms}
\crefname{algocf}{Algorithm}{Algorithms}
\Crefname{equation}{Inequality}{Inequalities}
\crefname{lstlisting}{listing}{listings}

\newcommand{\myemph}[1]{{\color{winered}\emph{#1}}}
\newcommand{\naturals}{{{\mathbb{N}}}}

\newcommand{\pos}{{{\mathrm{pos}}}}
\renewcommand{\top}{{{\mathrm{top}}}}

\newcommand{\extord}{\vartriangleright}
\newcommand{\dord}{\sqsupset}
\newcommand{\dordeq}{\sqsupseteq}

\newcommand{\low}{{{\mathrm{low}}}}
\newcommand{\grt}{{{\mathrm{grt}}}}

\renewcommand{\top}{{{\mathrm{top}}}}
\renewcommand{\bot}{{{\mathrm{bot}}}}

\renewcommand{\part}{{{\mathrm{part}}}}
\renewcommand{\H}{{{\mathrm{H}}}}

\newcommand{\first}{{{\mathrm{first}}}}

\theoremstyle{definition}
\newtheorem{definition}{Definition}
\newtheorem{example}{Example}
\newtheorem{lemma}{Lemma}
\newtheorem{observation}{Observation}
\theoremstyle{plain}
\newtheorem{theorem}{Theorem}
\newtheorem{corollary}{Corollary}

\newtheorem{question}{Question}
\newtheorem{proposition}{Proposition}

\makeatletter
\newtheorem*{rep@theorem}{\rep@title}
\newcommand{\newreptheorem}[2]{%
\newenvironment{rep#1}[1]{%
 \def\rep@title{#2 \ref{##1}}%
 \begin{rep@theorem}}%
 {\end{rep@theorem}}}
\makeatother
\newreptheorem{theorem}{Theorem}
\newreptheorem{proposition}{Proposition}
\newreptheorem{lemma}{Lemma}
\newreptheorem{corollary}{Corollary}

\begin{document}
	\title{Core-Stable Committees under Restricted Domains}
\author{Grzegorz Pierczyński\\
  University of Warsaw\\
  \href{mailto:g.pierczynski@mimuw.edu.pl}{g.pierczynski@mimuw.edu.pl}
  \and 
Piotr Skowron\\
  University of Warsaw\\
  \href{mailto:p.skowron@mimuw.edu.pl}{p.skowron@mimuw.edu.pl}
}
\date{}
	\maketitle

\begin{abstract}
We study the setting of committee elections, where a group of individuals needs to collectively select a given size subset of available objects. This model is relevant for a number of real-life scenarios including political elections, participatory budgeting, and facility-location. We focus on the core---the classic notion of proportionality, stability and fairness. We show that for a number of restricted domains including voter-interval, candidate-interval, single-peaked, and single-crossing preferences the core is non-empty and can be found in polynomial time. We show that the core might be empty for strict top-monotonic preferences, yet we introduce a relaxation of this class, which guarantees non-emptiness of the core. Our algorithms work both in the randomized and discrete models. We also show that the classic known proportional rules do not return committees from the core even for the most restrictive domains among those we consider (in particular for 1D-Euclidean preferences). We additionally prove a number of structural results that give better insights into the nature of some of the restricted domains, and which in particular give a better intuitive understanding of the class of top-monotonic preferences.
\end{abstract}

\lstset{
    keywords={input, output, for, while, if, else, return, break},
    comment=[l]{//},
    frame=single,
    mathescape=true,
    float,
    captionpos=b,
    numbers=left,
    breaklines=true,
}

\section{Introduction}

We consider a model of committee elections, where the goal is to select a fixed-size subset of objects based on the preferences of a group of individuals. The objects and the individuals are typically referred to as the \emph{candidates} and the \emph{voters}, respectively, and we follow this convention in our paper. However, the candidates do not need to represent  humans. For example, this model describes
\begin{inparaenum}[(1)]
\item the problem of locating public facilities---there the candidates correspond to possible physical locations where the facilities can be built~\citep{far-hek:b:facility-location,owaWinner}, 
\item the problem of presenting results by a search engine in response to a user query---there, the candidates are web-pages, and voters are potential users searching for a given query~\citep{proprank}, 
\item the problem of selecting validators in the blockchain, where the candidates are the users of the protocol~\citep{cevallos2020verifiably,burdges2020overview}.
\end{inparaenum}
For more examples that fall into the category of committee elections we refer to the recent book chapter~\citep{FSST-trends} and to the recent survey~\citep{lac-sko:abc-survey}.

In numerous applications that fit the model of committee elections it is critical to select a subset of candidates, hereinafter called a \emph{committee}, in a fair and proportional manner. Proportionality is one of fundamental requirements of methods for selecting representative bodies, such as parliaments, faculty boards, etc.  Yet, even in the context of facility location the properties that corresponds to proportionality are desirable: more objects should be built in densely populated area, ideally ensuring that the distribution of the locations of the built facilities resembles the distribution of the locations of the potential users. In case of searching, the returned results should contain items that are interesting to different types of users, or---in other words---the preferences of each minority of users should be represented in the returned set of results. Finally, the validators in the blockchain should proportionally represent the protocol users in order to make validation robust against coordinated attacks of malicious users. In all these examples it is important to select a proportional committee, yet it is not entirely clear what does it mean that the committee proportionally reflects the opinions of the voters, yet alone how to find such a committee.

The problem of formalizing the intuitive idea of proportionality has been often addressed in the literature and a plethora of axioms have been proposed  (see~\citep{lac-sko:abc-survey}[Section~5] and~\citep{FSST-trends}[Section~2.3.3]). Among them, the notion of the core is particularly interesting. This concept borrowed from the cooperative game theory~\citep{RePEc:mtp:titles:0262650401,chalkiadakis2011computational} can be intuitively described as follows. Assume our goal is to select a committee of $k$ candidates based on the ballots submitted by $n$ voters. Then, a group of $\nicefrac{n}{k}$ voters should intuitively have the right to select one committee member, and---analogously---a group of $\ell \cdot \nicefrac{n}{k}$ voters should be able to decide about the $\ell$ members of the elected committee. This intuition is formalized as follows: we say that a committee $W$ is in the core if no group $S$ of $\ell \cdot \nicefrac{n}{k}$ voters can propose a set of $\ell$ candidates $T$ such that each voter from $S$ prefers $T$ over $W$. 

The notion of the core is intuitively appealing, universal, and strong. It applies to different types of voters' ballots, in particular to the \emph{ordinal} and the \emph{approval} ones. In the ordinal model the voters rank the candidates from the most to the least preferred one, while in the approval model, the voters only mark the candidates that they find acceptable---we say that the voters approve such candidates.  Being in the core implies numerous other fairness-related properties, among them properties which are rather strong on their own. For example, in the approval model of voters' preferences the core implies the properties of extended justified representation (EJR)~\citep{justifiedRepresentation}, proportional justified representation (PJR)~\citep{AEHLSS18}, and justified representation (JR)~\citep{justifiedRepresentation}. For ordinal ballots, being in the core implies the properties of unanimity, consensus committee, and solid coalitions~\citep{elk-fal-sko-sli:c:multiwinner-rules}, as well as Dummett's proportionality~\citep{dum:b:voting} and proportionality for solid coalitions~(PSC)~\citep{az-bar:expanding_approval}. In \Cref{sec:max-core} we will also explain that for ordinal ballots the core is equivalent to full local stability~\citep{aziz2017condorcet}.

%Under approval ballots, the examples of properties that are worth mentioning are undoubtedly Justified Representation (JR) and Extended Justified Representation (EJR) proposed by \citep{justifiedRepresentation}, and Proportional Justified Representation (PJR) proposed by \citep{AEHLSS18}. These properties are known to be satisfiable and there exist natural voting rules (e.g. PAV) satisfying them for every election instance. 

%However, in case of one of the strongest, most natural and tempting concepts (i.e. implying all the ones mentioned above)---the core---these questions still remain open. For now, only approximation results (e.g. \citep{jiang2019approx}, \citep{cheng2019group}, \citep{pet-sko:laminar}) are known. Therefore, the goal of this work is to find algorithms satisfying the core.

While core-stability---the property of a voting rule that requires that each elected committee should belong to the core---is highly desired, it is also very demanding. For the ordinal ballots there exists no core-stable rule, and for the approval ballots it is one of the major open problems in computational social choice to find whether the property is satisfiable. Given the property is so demanding, so far the literature focussed on its relaxed versions---either the weaker properties which we mentioned before, or the approximate~\citep{jiang2019approx,pet-sko:laminar} and the randomized~\citep{cheng2019group} variants of the core.

In our work we explore a different, yet related approach. Our point is that before we look at how a voting rule works in the general case, at the very minimum we shall ensure that it behaves well on well-structured preferences. Thus, the main question that we state in this paper is whether the core-stability can be satisfied for certain natural restricted domains of voters' preferences, and what is the computational complexity of finding committees that belong to the core given elections where the voters' preferences come from restricted domains. The idea to restrict the scope only to instances in which the preferences are somehow well-structured is not new~\citep{elkind2017structured}, yet to the best of our knowledge it has never been considered in the context of the core. 

\subsection*{Our Contribution}

Our work contributes to two areas of computational social choice. First, we prove a number of structural theorems that describe existing domain restrictions. In particular, our results give a more intuitive explanation of the class of top-monotonic preferences. The original definition of this class is somewhat cumbersome. We show two independent conditions that provide alternative characterisations of top-monotonic preferences provided the voters' preference rankings have no ties. 
We also introduce two new domain restrictions which are natural, and which provide sufficient conditions for the existence of core-stable rules. One of our new classes generalizes voter-interval and candidate-interval domains~\citep{ijcai/ElkindL15-dichpref}, and the other class is a weakening of the domain of top-monotonic preferences; yet our class still includes single-peaked~\citep{bla:j:rationale-of-group-decision-making} and single-crossing preferences~\citep{mir:j:single-crossing,rob:j:tax}.

Second, we prove the existence of core-stable rules under the assumption that the voters' preferences come from certain restricted domains, in particular from domains of voter-interval, candidate-interval, single-peaked, and single-crossing preferences. Interestingly, we show a single algorithm that is core-stable for all four aforementioned domains. At the same time, we show that there exists no core-stable rule, even if we restrict our attention only to top-monotonic elections. 

The idea of our algorithm is the following. We first find a fractional (randomized) committee that is in the core. We pick those candidates that have been selected with probability equal to one. We choose the remaining candidates using a variant of the median rule applied to the truncated instance of the original election. Thus, our results hold both in the discrete and in the probabilistic case.

\section{Preliminaries}
%In this section we describe the input, the output and the goal of our model. 
For each $t\in \naturals$, we set $[t]=\{1, 2, \ldots, t\}$.

\subsection{Elections, Preferences, and Committtees}

An \myemph{election} is a tuple $E=(N, C, k)$, where $N=[n]$ is a set of $n$ \myemph{voters}, $C$ is a set of $m$ \myemph{candidates}, and $k$ is the desired \myemph{committee size}. 
%The voters have weak ordinal preferences; 
Each voter $i\in N$ submits her weak ranking $\succsim_i$ over the candidates---for each $i\in N$ and $a, b \in C$, we say that voter $i$ weakly prefers candidate $a$ over candidate $b$ if $a \succsim_i b$. We set $a \sim_i b$ if $a \succsim_i b$ and $b \succsim_i a$, and we write $a \succ_i b$ if $a \succsim_i b$ and $a \nsim_i b$.

For a voter $i\in N$ and $j\in [m]$, by $\pos_i(j)$ we denote the equivalence class of candidates ranked on the $j$-th position by voter $i$. Formally, a candidate $c$ belongs to $\pos_i(j)$ if there are $(j-1)$ candidates $a_1, \ldots, a_{j-1}$ such that $a_1 \succ_i a_2 \succ_i \ldots \succ_i a_{j-1} \succ_i c$ and if there exist no $j$ candidates $a_1, \ldots, a_{j}$ for which  $a_1 \succ_i a_2 \succ_i \ldots \succ_i a_{j} \succ_i c$. By $d_i$ we denote the number of the nonempty positions in the $i$-th voter's preference list. For each $j\in [d_i]$, by $\pos_i([j])$ we denote $\bigcup_{q \leq j} \pos_i(q)$, and for convenience, we set $\top_i = \pos_i(1)$ and $\bot_i = \pos_i(d_i)$.

We distinguish two specific types voters' preferences.
\begin{description}
    \item[Approval preferences.] We say that the preferences are \myemph{approval}, if for each voter $i\in N$ it holds that $d_i \leq 2$. We say that $i$ \myemph{approves} $c$ if $c\in \top_i$.
    \item[Strict preferences.] The voters' preferences are \myemph{strict}, if $d_i = m$ for each $i\in N$.
\end{description}

We call $k$-element subsets of $C$ size-$k$ \myemph{committees}, or in short committees, if the size $k$ is clear from the context. We extend this notion to the continuous model as follows.
A \myemph{fractional committee} is a function $p\colon C \mapsto [0;1]$ that assigns to each candidate from $c \in C$ a value $p(c)$ such that $0 \leq p(c) \leq 1$; intuitively $p(c)$ can be thought of as the probability that candidate $c$ is a member of the selected committee. We extend this notation to sets, defining $p(T) = \sum_{c \in T} p(c)$ for each $T \subseteq C$. The value of $p(C)$ is the \myemph{size} of the fractional committee. If for a candidate $c$ it holds that $p(c)=1$, then we say that $c$ is \myemph{elected}, otherwise she is \myemph{unelected}. If for an unelected candidate $c$ it holds that $p(c) > 0$, then $c$ is \myemph{partially elected}. If there are no partially elected candidates in $p$, then we say that $p$ is a \myemph{discrete committee} (or simply a committee) and associate it with the subset of candidates $\{c\in C\colon p(c) = 1\}$.

A \myemph{voting rule} is an algorithm that takes as input an election, and returns a nonempty set of committees, hereinafter called winning committees.\footnote{Typically a voting rule would return a single winning committee, but ties are possible.} A fractional voting rule is an algorithm that given an election returns a fractional committee.

The notion of a fractional committee is similar to several probabilistic concepts considered in the literature. For instance, in probabilistic social choice (see a book chapter by Brandt~\citep{Bran17a}) we also assign fractional values to candidates. The main difference is that in probabilistic social choice, the whole value that we want to divide among the candidates can be assigned to fewer than $k$ candidates; in particular it is feasible to set $p(c) = k$ for one candidate and $p(c') = 0$ for all $c'$, $c' \neq c$. Thus, intuitively, in probabilistic social choice each candidate is divisible and appears in an unlimited quantity. Viewing from this perspective, probabilistic social choice extends the discrete model of approval-based apportionment~\citep{bri-got-pet-sch-wil:approval-apportionment}. Several works have considered axioms of proportionality for probabilistic social choice~\citep{ABM19:fair-mixing,FGM16a}, yet unfortunately their results do not apply to fractional committees.
 
Another concept related to fractional committees is where we assign probabilities to committees instead of individual candidates. The notions of proportionality in this setting have been considered, e.g., by Cheng~et~al.~\citep{cheng2019group}.  
It is worth noting that fractional committees can induce probability distributions over committees, e.g., by applying sampling techniques, such as dependent rounding~\citep{sr:dependent-rounding}, that ensure we always select $k$ candidates. Yet, there is no one-to-one equivalence between the two settings, thus the results of Cheng~et~al.~\citep{cheng2019group} do not apply to fractional committees.

\subsection{The Core as a Concept of Proportionality}\label{sec:the-core}

There are numerous axioms that aim at formalizing the intuitive idea of proportionality. In this paper we focus on one of the strongest such properties, the \myemph{core}~\cite{justifiedRepresentation}. The idea behind the definition of the core is the following: a group of voters $S$ shall be allowed to decide about a subset of candidates that is proportional to the size of $S$; for example a group consisting of 70\% of voters should have the right to decide about 70\% of the elected candidates. The core prohibits situations where a group $S$ can propose a proportionally smaller set of candidates $T$ such that each voter from $S$ would prefer $T$ to the committee at hand.

%The classic definition of the core is formulated only for discrete committees and looks as follows:

\begin{definition}[The core]\label{def:core}
Given an election instance $E=(N, C, k)$, we say that a committee $W$ is in the \myemph{core}, if for each $S \subseteq N$ and each subset of candidates $T$ with $|T| \leq k \cdot \nicefrac{|S|}{n}$ there is a voter $v \in S$ weakly preferring $W$ to $T$.
\end{definition}

In the above definition we still need to specify how the voters' compare committees, that is how their preferences over individual candidates can be extended to the preferences over committees. Specifically, for each voter $i \in N$ by $\extord_i$ we denote the partial order over $2^C$ being the result of extending the preference relation of a voter $i$.

Through the whole paper, except for \Cref{sec:max-core}, we use the \myemph{lexicographic} extension:
\begin{align}\label{eq:def-ext-lex}
\begin{split}
W \extord_i T \iff \exists \sigma\in [d_i].\ &|\pos_i(\sigma) \cap W| > |\pos_i(\sigma) \cap T|
\\&\text{ and }\forall \varrho < \sigma.\ |\pos_i(\varrho) \cap W| = |\pos_i(\varrho) \cap T|\text{.}
\end{split}
\end{align}
Note that for approval preferences the lexicographic extension boils down to counting approved candidates in $T$ and $T'$. Since in \Cref{def:core} $|W| \geq |T|$, a voter weakly prefers $W$ over $T$ whenever she approves at least as many candidates in $W$ as in $T$.

Below we generalize the preference extension \eqref{eq:def-ext-lex} to fractional committees: 
\begin{align*}
\begin{split}
p \extord_i p' \iff \exists \sigma\in [d_i].\ &p(\pos_i(\sigma)) > p'(\pos_i(\sigma))
\\&\text{ and }\forall \varrho < \sigma.\ p(\pos_i(\varrho)) = p'(\pos_i(\varrho))\text{.}
\end{split}
\end{align*}

%Once we defined the preference extensions for fractional committees, 
\Cref{def:core} naturally extends to fractional committees. 

\begin{definition}[The core (for fractional committees)]\label{def:fractional-core}
Given an election instance $E=(N, C, k)$, we say that a fractional committee $p$ is in the core, if for each $S \subseteq N$ and each fractional subset $p'$ with $p'(C) \leq k \cdot \nicefrac{|S|}{n}$, there exists a voter $i \in S$ such that $i$ weakly prefers $p$ over $p'$.
\end{definition}

We say that a voting rule is core-stable if it always returns committees in the core.

\section{Restricted domains}\label{sec:restricted_domains}

A voting rule specifies an outcome of an election independently of how the voters' preferences look like.  Similarly, the core puts certain structural requirements on the selected committees that should be satisfied in every possible election. However, the space of all elections is reach and it might be too demanding to expect a voting rule to satisfy a strong property in each possible case. For example, this is the case for the core: there are elections with strict rankings where no committee belongs to the core~\cite{fain2018fair,cheng2019group}; the question whether the core is always non-empty assuming approval preferences is still open. Instead, what is often desired is that a voting rule should satisfy strong notions of proportionality when the voters' preferences are in some sense logically consistent. This motivates focusing primarily on election instances where the voters' preferences are well-structured, or---in other words---come from certain restricted domains.

In this section we describe a few known and introduce one new preference domain. We show that the commonly known voting rules are not core-stable even for the most restricted preference domains. We also provide alternative conditions characterising some of the considered domains. These results will help us in our further analysis of voting methods, but are also interesting on their own.

%The problem of the emptiness of the core remains unsolved for elections with unrestricted preferences, even in the approval model. Hence the natural idea is to assume that voters' preferences are somehow structured and study the issue of the existence of the core in so called \emph{restricted domains}. The best-known of them are defined only for approval or strict preferences, hence we discuss them separately in next subsections.

\subsection{Strict preferences}

For strict ordinal preferences we start by recalling the definitions of the following two known preference classes.

\begin{definition}[Single-crossing preferences]\label{single-crossing}
Given an election instance $E=(N, C, k)$, we say that $E$ has single-crossing preferences if there exists a linear order $\dord$ over voters such that for each voters $x \dord y \dord z$ and candidates $a, b \in C$ such that $a \succ_y b$ we have that $b \succ_x a \implies a \succ_z b$.
\end{definition}

Intuitively, we say that preferences are single-crossing if the voters can be ordered in such a way that for each pair of candidates, $a, b \in C$, the relative order between $a$ and $b$ changes at most once while we move along the voters.

\begin{definition}[Single-peaked preferences]\label{single-peaked}
Given an election instance $E=(N, C, k)$, we say that $E$ has single-peaked preferences if there exists a linear order $\dord$ over candidates such that for each voter $i \in N$ and candidates $a \dord b \dord c$ we have that $\top_i = a \implies b \succ_i c$ and $\top_i = c \implies b \succ_i a$.
\end{definition}

\begin{definition}[1D-Euclidean preferences]
Given an election instance $E=(N, C, k)$, we say that $E$ has 1D-Euclidean preferences if there exists a 1D-Euclidean metric space in which both voters and candidates are located, such that each voter $i\in N$ prefers a candidate $a$ to a candidate $b$ if and only if $a$ is closer to $i$ than $b$.
\end{definition}

Every 1D-Euclidean election is both single-peaked and single-crossing. 

\subsubsection{Known Voting Rules are not Core-Stable for 1D-Euclidean Preferences}

To the best of our knowledge, none of the known voting rules is core-stable, even for 1D-Euclidean elections. We prove this for two archetypal proportional rules, the Monroe's rule and STV.

\begin{definition}[The Monroe Rule]
Consider an election $E$ with strict preferences and assume that $\nicefrac{n}{k}$ is integral. For a committee $T\subseteq C$, a \emph{balanced matching} is a collection of subsets of voters $\{N_c\}_{c\in T}$ such that for every $c\in T$, $|N_c| = \nicefrac{n}{k}$. The \emph{value} of a matching is the sum $\sum_{c\in T} \sum_{i\in N_c} \pos_i(\{c\})$. The matching $\{N_c\}_{c\in T}$ is minimal, if it has the minimal value among all matchings for $T$. The Monroe Rule returns the committee $W$ minimizing the value of the minimal balanced matching.
\end{definition}

\begin{definition}[Single Transferable Vote (STV)]
Consider an election $E$ with strict preferences. STV proceeds sequentially: at each round we elect a candidate that is ranked top by at least $\nicefrac{n}{k+1}+1$ voters and remove any $\nicefrac{n}{k+1}+1$ of these voters from the election. If there are no such candidates, we remove from the election a candidate ranked top by the least number of voters.
\end{definition}

Both the Monroe Rule and STV are not core-stable even for 1D-Euclidean instances, as shown in \Cref{ex:monroe} and \Cref{ex:stv}, respectively.

\begin{example}\label{ex:monroe}
Let $k=2$. Voters' preferences are the following:
\begin{align*}
v_1\colon b \succ a \succ c \succ d \succ e\\
v_2\colon c \succ b \succ d \succ a \succ e\\
v_3\colon c \succ d \succ b \succ e \succ a\\
v_4\colon d \succ e \succ c \succ b \succ a
\end{align*}
This instance is 1D-Euclidean as presented in \Cref{fig:ex:monroe}. 

Here the committee $\{b, d\}$ is elected by Monroe, but group $S$ consisting of middle voters $\{2, 3\}$ and $T=\{c\}$ witness the core violation.
\end{example}

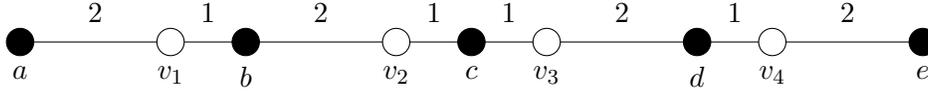
\begin{figure}
    \centering
    \begin{tikzpicture}
    \draw (0,0) -- (12,0);
    
    \node[label=above:$2$] at (1,0) {};
    \node[label=above:$1$] at (2.5,0) {};
    \node[label=above:$2$] at (4,0) {};
    \node[label=above:$1$] at (5.5,0) {};
    \node[label=above:$1$] at (6.5,0) {};
    \node[label=above:$2$] at (8,0) {};
    \node[label=above:$1$] at (9.5,0) {};
    \node[label=above:$2$] at (11,0) {};

    \node[circle, draw=black, fill=black, label=below:$a$] at (0,0) {};
    \node[circle, draw=black, fill=white, label=below:$v_1$] at (2,0) {};
    \node[circle, draw=black, fill=black, label=below:$b$] at (3,0) {};
    \node[circle, draw=black, fill=white, label=below:$v_2$] at (5,0) {};
    \node[circle, draw=black, fill=black, label=below:$c$] at (6,0) {};
    \node[circle, draw=black, fill=white, label=below:$v_3$] at (7,0) {};
    \node[circle, draw=black, fill=black, label=below:$d$] at (9,0) {};
    \node[circle, draw=black, fill=white, label=below:$v_4$] at (10,0) {};
    \node[circle, draw=black, fill=black, label=below:$e$] at (12,0) {};
    \end{tikzpicture}
    \caption{An illustration of \Cref{ex:monroe}. White and black points mean the positions of respecitvely the voters and the candidates.}
    \label{fig:ex:monroe}
\end{figure}

\begin{example}\label{ex:stv}
Let $n=60$, $k=2$. The value of the STV quota is $\nicefrac{n}{k+1}+1=21$. Voters' preferences are divided into 5 groups:
\begin{alignat*}{2}
&G_1\text{ }(18~\text{voters})\colon& a \succ b \succ c \succ d \succ e\\
&G_2\text{ }(7~\text{voters})\colon& b \succ c \succ d \succ e \succ a\\
&G_3\text{ }(5~\text{voters})\colon& c \succ d \succ e \succ b \succ a\\
&G_4\text{ }(16~\text{voters})\colon& d \succ e \succ c \succ b \succ a\\
&G_5\text{ }(14~\text{voters})\colon& e \succ d \succ c \succ b \succ a
\end{alignat*}
This instance is 1D-Euclidean as presented in \Cref{fig:ex:stv}. 

Here candidate $c$ is eliminated at the first round and all votes for her are transferred to $d$. Second, candidate $d$ is elected (in the second round she gains exactly $21$ votes) and the votes from groups $3$ and $4$ are removed. Third, candidate $b$ is eliminated and all votes for her are transferred to $e$. Fourth, candidate $e$ is elected (gaining in the final round exactly $21$ votes) and the committee $\{d, e\}$ is returned. However, $30$ voters from the three first groups and candidate $c$ witness the core violation.
\end{example}

\begin{figure}
    \centering
    \begin{tikzpicture}
    \draw (0,0) -- (15,0);
    \node[circle, draw=black, fill=black, label={[align=center]below:$G_1$\\$a$}] at (0,0) {};
    \node[circle, draw=black, fill=black, label={[align=center]below:$G_2$\\$b$}] at (8,0) {};
    \node[circle, draw=black, fill=black, label={[align=center]below:$G_3$\\$c$}] at (12,0) {};
    \node[circle, draw=black, fill=black, label={[align=center]below:$G_4$\\$d$}] at (14,0) {};
    \node[circle, draw=black, fill=black, label={[align=center]below:$G_5$\\$e$}] at (15,0) {};
    
    \node[label=above:$8$] at (4,0) {};
    \node[label=above:$4$] at (10,0) {};
    \node[label=above:$2$] at (13,0) {};
    \node[label=above:$1$] at (14.5,0) {};
    
    \end{tikzpicture}
    \caption{An illustration of \Cref{ex:stv}. Black points mean the positions of both candidates and groups of voters (ties can be broken arbitrarily).}
    \label{fig:ex:stv}
\end{figure}
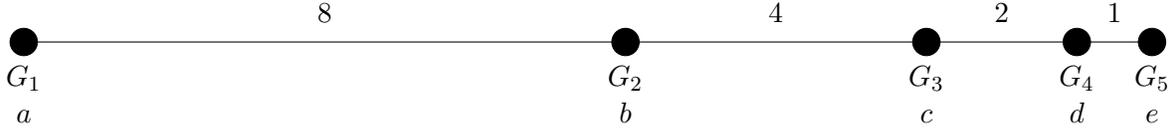

\subsubsection{Top Monotonic Preferences}

%On the other hand, in \Cref{sec:algorithm-strict} we prove that for single-peaked and single-crossing elections the core is always non-empty. In the proof we would not refer directly to these classes---instead we propose a new restricted domain, 

We will now recall the definition of the \myemph{top monotonic} domain~\citep{BARBERA2011345}. This domain is defined assuming the voters submit their preferences as weak orders. For strict preferences it  generalizes both the single-peaked and single-crossing domains. We call all candidates that are ranked in the position by at least one voter top-candidates. 
%It is a modification of the known class of \myemph{top monotonic} preferences (let us denote by \emph{top candidates} the candidates who are ranked top by at least one voter):

\begin{definition}[Top monotonicity (TM)]\label{def:top_monotonicity}
Given an election $E=(N, C, k)$, we say that $E$ has \myemph{top monotonic} preferences if there exists a linear order $\dord$ over candidates such that the two following conditions hold:
\begin{itemize}
    \item for each candidates $a, b, c$ and voters $i, j$ such that $a\in \top_i$ and $b\in \top_j$ it holds that:
    \begin{align*}
    \begin{split}
        (a \dord b \dord c \text{ or } c \dord b \dord a) \implies &b \succsim_i c \text{ if } c \in \top_i \cup \top_j\\
        &b \succ_i c \text{ otherwise}
    \end{split}
    \end{align*}
    \item the same implication holds also for each top candidates $a, b, c$ and voters $i, j$ such that $a \succsim_i b, c$ and $b \succsim_j a, c$.
\end{itemize}
\end{definition}

The definition of top-monotonic preferences is complex and somewhat counterintuitive. We will first show that for strict orders this definition can be equivalently characterized by two much simpler and more intuitive conditions.

\begin{definition}[Single-top-peaked (STP) preferences]
Given an election $E=(N, C, k)$, we say that $E$ has \myemph{single-top-peaked} preferences if there exists a linear order $\dord$ over candidates such that for each candidates $a \dord b \dord c$ such that $b$ is a top candidate, and a voter $i$ it holds that $\top_i = a \implies b \succ_i c$ and $\top_i = c \implies b \succ_i a$.
\end{definition}

\begin{proposition}
In the strict model, single-top-peakedness is equivalent to top-monotonicity.
\end{proposition}
\begin{proof}
Observe that the first condition in the definition of  TM implies STP. Now, we will show the reverse implication.
Consider an STP election. We will show that it satisfies the two conditions specified in \Cref{def:top_monotonicity}.

Note that in the strict model if the premise of the first condition is satisfied, then $\top_i = \{a\}$ and $\top_j = \{b\}$ and $c\notin \top_i \cup \top_j$. Hence, the first condition follows from the condition for STP.

Consider now the second condition. If $a\succsim_i b, c$ and $b\succsim_j a, c$, then in the strict model it holds that $a\succ_i b, c$ and $b\succ_j a, c$. Let us consider two cases: first assume that $a\dord b \dord c$. We know that $\top_i=\{d\}$ for some $d\in C\setminus\{b, c\}$. If $d \dord b$, then $b\succ_i c$ follows from the definition of STC (for voter $i$ and candidates $d, b,c$). Suppose now that $b \dord d$. But then from the definition of STP (for voter $i$ and candidates $a, b, d$) we obtain that $b \succ_i a$, a contradiction. The reasoning for the case when $c\dord b \dord a$ is analogous.
\end{proof}

It is clear that the definition of STP is closely connected to the definition of single-peaked preferences (only the condition is partially weakened to the candidates that are ranked top by some voter). One could also consider the analogous weakening for single-crossing preferences.

\begin{definition}[Single-top-crossing (STC) preferences]
Given an election $E=(N, C, k)$, we say that $E$ has \myemph{single-top-crossing} preferences if there exists a linear order $\dord$ over voters such that for each voters $x \dord y \dord z$ and a candidate $a\in C$, we have that $a \succ_x \top_y \implies \top_y \succ_z a$.
\end{definition}

Although the definitions of STC and STP look different, they are in fact equivalent.

\begin{proposition}\label{prop:stc_and_stp}
In the strict model, single-top-peakedness is equivalent to single-top-crossingness.
\end{proposition}
\begin{proof}
Consider an STC election $E$ and a linear order $\dord$ over voters given by the definition of STC. We say that $i$ preceds $j$ if $j \dord i$. We construct the linear order over candidates as follows:
\begin{enumerate}
    \item Consider some $a, b \in C$ such that $a$ is the top preference for some voter $i\in N$. From the definition of STC, we know that voters preferring $b$ to $a$ can all either succeed or precede~$i$. If they succeed $i$, then we add constraint $b \dord a$, otherwise we add constraint $a \dord b$. If there are no voters prefering $b$ to $a$, we add no constraint. We repeat this step for each pairs $a, b \in C$.
    \item Finally, if after the previous step some pairs are still uncomparable, we complete the order in any transitive way.
\end{enumerate}

We will show that the constraints placed during the first step of the procedure are transitive. Indeed, consider (for the sake of contradiction) three candidates $a, b, c$ such that the procedure placed constraints $a \dord b$, $b \dord c$ and $c \dord a$. Hence, we know that at least two out of these three candidates are top candidates. Assume without the loss of generality that $a$ and $b$ are top candidates. Let $i_a, i_b$ be voters ranking top respectively $a$ and $b$ (naturally, $i_a \dord i_b$). We know that all the voters preceding $i_a$ prefer $a$ to $c$ and all the voters preceding $i_b$ prefer $c$ to $b$. There exists at least one voter $i$ preferring $c$ to $b$ (as otherwise constraint $b \dord c$ would not be added) and $i_b \dord i$. By transitivity of the preference relation, we know that $i$ prefers $a$ over $b$. Consequently, $i_a, i_b$ and $i$ together with candidate $a$ witness STC violation. The obtained contradiction shows that the order $\dord$ is indeed transitive.

We will now prove that such linear order $\dord$ over candidates satisfies the conditions of STP. Indeed, consider any three candidates $a \dord b \dord c$ such that $b$ is a top candidate and a voter $i\in N$. Let $\top_i = \{a\}$. As $b$ is a top candidate, there exists a voter $j$ such that $\top_j=\{b\}$. As $a \dord b$, it holds that $i \dord j$. Then if we had that $c \succ_i b$, our procedure would place constraint $c \dord b$, a contradiction. Hence $b \succ_i c$. The proof for the case $\top_i = \{c\}$ is analogous.

Now we will prove the reverse implication. Let $E$ be an STP election with a linear order $\dord$ over the candidates. Consider the following linear order $\dord$ over the voters: for each $i, j\in N$ we have that if $\top_i \dord \top_j$ then $i \dord j$.  Now consider three voters $x, y, z$ and a candidate $a$ such that $a \succ_x \top_y$. Suppose that $a \dord \top_y$. Then from the properties of top monotonocity and the fact that $\top_y \dord \top_z$, we have that $z$ has preference ranking $\top_z \succ_z \top_y \succ_z a$. Suppose now that $\top_y \dord a$. But since $\top_x \dord \top_y \dord a$, the fact that $a \succ_x \top_y$ leads to the contradiction with the definition of top monotonicity, which completes the proof.
\end{proof}

Recall that single-crossingness implies single-peakedness for narcissist domains, i.e., under the assumption that each candidate is ranked top at least once~\cite{elk-fal-sko:characterizationOfSingleCrossing}. Since for narcissist domains a single peaked profile is also single-top-peaked, we get a related result: that single-peakedness is equivalent to single-top-crossingness assuming narcissist preferences.

The class of top monotonic preferences (TM) puts a focus on the top positions in the voters' preference rankings. For example, an election in which the voters unanimously rank a single candidate as their most preferred choice is top-monotonic, independently of how the other candidates are ranked. This suggests that TM offers a combinatorial structure that might be useful in the analysis of single-winner elections, but which might not help to reason about committees. Indeed, below we define a new class which is a natural strengthening of TM. In \Cref{sec:algorithm} we show that the core is always nonempty for elections belonging to our newly defined class, and we show that this is not the case for the original class of TM.

\begin{definition}[Recursive single-top-crossing (r-STC) preferences]
Given an election $E=(N, C, k)$, we say that $E$ has \myemph{recursive single-top-crossing} preferences if every subinstance of $E$ obtained by removing some candidates from $E$ is STC.
\end{definition}

Although r-STC is stricter than STC, it still contains both single-peaked and single-crossing preferences. This follows from the fact that both single peaked and single-crossing preferences are top monotonic~\citep{BARBERA2011345}, and that single-peakedness and single-crossingness is preserved under the operation of removing candidates from the election.

\subsection{Approval elections}

In the approval model we first recall the definitions of two classic domain restrictions, the voter-interval and the candidate-interval models~\cite{ijcai/ElkindL15-dichpref}.

\begin{definition}[Voter-interval (VI) preferences]\label{voter-interval}
Given an election instance $E=(N, C, k)$, we say that $E$ has \myemph{voter-interval} preferences if there exists a linear order $\dord$ over $N$ such that for all voters $v_1, v_2, v_3\in N$ and for each candidate $c \in \top_{v_1} \cap \top_{v_3}$, we have that $v_1 \sqsupset v_2 \sqsupset v_3 \implies c \in \top_{v_2}$. Intuitively, each candidate is approved by a consistent interval of voters.
\end{definition}

\begin{definition}[Candidate-interval (CI) preferences]\label{candidate-interval}
Given an election instance $E=(N, C, k)$, we say that $E$ has \myemph{candidate-interval} preferences if there exists a linear order $\dord$ over $C$ such that for each voter $i \in N$ and all candidates $a, c \in \top_i, b \in C$ we have that $a \dord b \dord c \implies b \in \top_i$. Intuitively, each voter approves a consistent interval of candidates.
\end{definition}

\subsubsection{Known Voting Rules are not Core-Stable for VI nor CI Preferences}

As in the case of strict preferences, we first show that two known proportional voting rules are not core stable even if the preferences come from the above restricted domains. We focus on known rules that satisfy extended justified representation (EJR)~\cite{justifiedRepresentation}, one of the strongest proportionality axioms that are known to be satisfiable in general.

\begin{definition}[Proportional Approval Voting (PAV)~\cite{Thie95a}]\label{def:pav}
Given an election instance $E=(N, C, k)$, we elect a committee $W$ maximizing the value of the following expression:
\begin{equation*}
    \sum_{i\in N}\H(|W \cap \top_i|)  \quad \text{where} \quad \H(i) = 1 + \frac{1}{2} + \ldots + \frac{1}{i} \text{.}
\end{equation*}
\end{definition}

\begin{definition}[Rule X~\cite{pet-sko:laminar}]\label{def:rulex}
We assume that each voter is given $1$ dollar at the beginning. Every candidate needs to be paid $\nicefrac{n}{k}$ dollars to be elected. The algorithm is sequential. At each round, we iterate over candidates and for each candidate $c$ we compute the value $\varrho_c$---the lowest value such that the voters approving $c$ can afford her election (i.e., can afford to pay $\nicefrac{n}{k}$ dollars in total) provided each voter pays at most $\varrho_c$. Then we elect the affordable candidate minimising $\varrho_c$, decrease the voters' budgets and repeat the procedure until there are no affordable candidates.
\end{definition}

%The choice of these particular rules was reasonable---they are the only rules known to satisfy the property called \emph{extended justified representation} (formulated by \citep{justifiedRepresentation}), which is the weakening of the core and one of the strongest proportionality axioms that are known to be satisfiable in general.

Both these rules are not core-stable even for elections belonging to the intersection of VI and CI classes, as shown in \Cref{ex:pav} and \Cref{ex:rulex}.

\begin{example}\label{ex:pav}
Let $n=3$, $k=8$. Voters' preferences are the following:
\begin{align*}
& v_1\colon\{b_1, b_2, b_3, b_4, a\}\\
& v_2\colon\{b_1, b_2, b_3, b_4, c\}\\
& v_3\colon\{d_1, d_2, d_3, d_4\}
\end{align*}
Assuming
\begin{align*}
    v_1 \dord v_2 \dord v_3
\end{align*}
and
\begin{align*}
    a \dord b_1 \dord \ldots \dord b_4 \dord c \dord d_1 \dord \ldots \dord d_4,
\end{align*}
it is clear that the instance is both VI and CI.

Here PAV elects candidates $\{b_1, \ldots, b_4, d_1, \ldots, d_4\}$. However, this committee does not belong to the core, which is witnessed by the groups $S=\{v_1, v_2\}$ and $T=\{a, b_1, \ldots, b_4, c\}$.
\end{example}

\begin{example}\label{ex:rulex}
Let $n=42$, $k=14$. Voters' preferences are divided into the following groups:
\begin{alignat*}{2}
&G_1 \text{ }(1~\text{voter})  \colon&& \{c_1, c_2, c_3 , x_1, x_2\}\\
&G_2 \text{ }(8~\text{voters}) \colon&& \{c_1, c_2, c_3, x_1, x_2, a_1, a_2, a_3, a_4\}\\
&G_3 \text{ }(12~\text{voters}) \colon&& \{c_1, c_2, c_3, a_1, a_2, a_3, a_4, b_1, b_2, b_3, b_4, e_1, e_2\}\\
&G_4 \text{ }(12~\text{voters}) \colon&& \{d_1, d_2, d_3, b_1, b_2, b_3, b_4, a_1, a_2, a_3, a_4, e_1, e_2\}\\
&G_5 \text{ }(8~\text{voters}) \colon&& \{d_1, d_2, d_3, y_1, y_2, b_1, b_2, b_3, b_4\}\\
&G_6 \text{ }(1~\text{voter}) \colon&& \{d_1, d_2, d_3, y_1, y_2\}
\end{alignat*}
Assuming
\begin{align*}
    G_1 \dord G_2 \dord \ldots \dord G_6 \text{ (the voters within each group can be ordered arbitrarily)}
\end{align*}
and
\begin{align*}
    x_1 \dord x_2 \dord c_1 \dord c_2 \dord c_3 \dord a_1 \dord \ldots \dord a_4 \dord e_1 \dord e_2 \dord b_1 \dord \ldots \dord b_4 \dord d_1 \dord d_2 \dord d_3 \dord y_1 \dord y_2 
\end{align*}
it is clear that the instance is both VI and CI.

At the beginning each voter has $1$ dollar and the price for candidates is $p=\nicefrac{n}{k} = 3$. First we elect candidates $a_1, \ldots, a_4, b_1, \ldots, b_4$; for each of them $32$ out of $40$ middle voters pay (each of them pays~$\nicefrac{3}{32}$). Second, we elect candidates $e_1$ and $e_2$, and the middle $24$ voters run out of money; indeed, each of them pays $8 \cdot \nicefrac{3}{32} + 2 \cdot \nicefrac{3}{24} = 1$ for the so far elected candidates. Next we elect candidates $x_1, x_2, y_1, y_2$. We have elected the committee $\{a_1, \ldots, a_4, b_1, \ldots, b_4, e_1, e_2, x_1, x_2, y_1, y_2\}$ which is not even Pareto-optimal as the committee $\{a_1, \ldots, a_4, b_1, \ldots, b_4, c_1, c_2, c_3, d_1, d_2, d_3\}$ is better for every voter. Thus, in particular the elected committee does not belong to the core.
\end{example}

\subsubsection*{Linearly consistent (LC) preferences}

%On the other hand, in \Cref{sec:algorithm-app} we prove that for CI and VI elections the core is always non-empty. In the proof we would not refer directly to these classes---instead we propose a new restricted domain, generalizing both CI and VI preferences.

Below we introduce a new class that generalizes both CI and VI domains. In \Cref{sec:algorithm-app} we will prove that the core is always nonempty if preferences come from our new restricted domain.

\begin{definition}[Linearly consistent (LC) preferences]\label{def:lc}
Given an election instance $E=(N, C, k)$, we say that $E$ has linearly consistent preferences, if there exists a linear order $\dord$ over $N\cup C$ such that for each voters $i, j \in N$ ($i \dord j$) and candidates $a, b\in C$ ($a \dord b$), if $a \in \top_j$, then $a \succsim_i b$. In words, if $i$ approves $b$ and $j$ approves $a$, then $i$ approves $a$ (this intuition is depicted in \Cref{fig:lc}).
\end{definition} 

\begin{figure}
    \centering
    \begin{tikzpicture}
    \node at (-2, 0) {candidates:};
    \node (A) at (0,0) {$a$};
    \node at (2,0) {$\dord$};
    \node (B) at (4,0) {$b$};
    
    \node at (-2, 5) {voters:};
    \node (I) at (0,5) {$i$};
    \node at (2,5) {$\dord$};
    \node (J) at (4,5) {$j$};
    
    \draw [-{Stealth[scale=2]}, dashed] (I) -- (A) node [pos=.3, below, sloped] {\scriptsize{needs to approve}};
    \draw [-{Stealth[scale=2]}] (I) -- (B) node [pos=.3, below, sloped] {\scriptsize{approves}};
    \draw [-{Stealth[scale=2]}] (J) -- (A) node [pos=.3, below, sloped] {\scriptsize{approves}};
    
    \end{tikzpicture}
    \caption{An illustration of the definition of LC.}
    \label{fig:lc}
\end{figure}
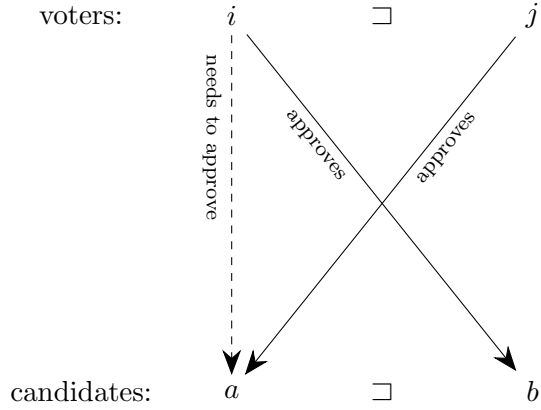

\begin{proposition}\label{obs:ci-vi-implies-lc}
Each VI election is LC. Each CI election is LC.
\end{proposition}
\begin{proof}
The case of voter-interval preferences. Let $\dord$ be a linear order over $N$ that witnesses that preferences are voter-interval. Let us sort $N$ by this order. For each candidate $c$, by $\first_c$ we denote $\min \{i\in N\colon c\in \top_i\}$. Let us now associate each candidate $c$ to $\first_c$ (breaking the tie between $c$ and $\first_c$ arbitrarily). If two candidates $a, b$ are associated to the same point, we also break the tie between them arbitrarily. In such a way we obtained an order $\dord$ over $N \cup C$.  For simplicity, for each $x, y\in N\cup C$ by $x\dordeq y$ we denote "$x \dord y$ or $x=y$".

Consider two voters, $i$ and $j$, with $i \dord j$, and two candidates, $a$ and $b$, with $a \dord b$. Assume $i$ approves $b$ and $j$ approves $a$. We will prove that $i$ approves $a$. Since $a \dord b$, by our definition $\first_a \dordeq \first_b$. Since $i$ approves $b$, $\first_b \dordeq i$, and so $\first_a \dordeq i$. If $i = \first_a$, then $i$ approves $a$. Otherwise, $\first_a \dord i$. Consequently, $\first_a$, $i$, and $j$ are three voters, such that $\first_a \dord i \dord j$. Since the preferences are voter-interval we infer that $i$ approves $a$.

The case of candidate-interval preferences. Let $\dord$ be a linear order on $C$ witnessing the candidate-interval property. Let us sort $C$ by this order. We associate each voter $i\in N$ with $(\min \top_i)$, again breaking all the ties arbitrarily. Consider  two voters, $i$ and $j$ with $i \dord j$, and two candidates $a$ and $b$, with $a \dord b$. Further, assume that $i$ approves $b$ and $j$ approves $a$. Since $i \dord j$, we get that $(\min \top_i) \dordeq (\min \top_j)$, and since $j$ approves $a$, we have $(\min \top_j) \dordeq a$.  Consequently, $(\min \top_i) \dordeq a$. If $(\min \top_i) = a$, then $i$ approves $a$. Otherwise, $(\min \top_i)$, $a$ and $b$ are three candidates, such that $(\min \top_i) \dord a \dord b$. Given that preferences are candidate-interval, and that $i$ approves $b$, we get that $i$ approves $a$.
\end{proof}

Let us now compare the domain of linearly consistent preferences with the one of \emph{seemingly single-crossing (SSC)} preferences~\citep{elkind2017structured}---another known class that generalizes VI and CI domains. We say that preferences are seemingly single-crossing if there is a linear order over voters such that for each $a, b\in C$, the voters approving $a$ and not $b$ either all succeed or all precede the voters approving $b$ and not $a$.\footnote{There is also another class, generalising both VI and CI---namely, the class of \emph{possibly single-crossing} (PSC) preferences. This is the class of approval preferences that can be obtained from some strict single-crossing ones (assuming that every voter approves a consistent prefix of her ranking). Interestingly, PSC is equivalent  to the class of seemingly single-crossing preferences~\cite{elkind2017structured}.} Observe that LC implies SSC. The reverse implication does not hold, as we show in \Cref{ex:ssc-not-lc} below.

\begin{example}\label{ex:ssc-not-lc}
Consider the election instance with $3$ voters and the following preferences:
\begin{align*}
    v_1\colon \{a, c\}\\
    v_2\colon \{a, b\}\\
    v_3\colon \{b, c\}
\end{align*}
It is straightforward to check that these preferences are SSC for all pairs of candidates and any linear order over voters.

Suppose that this election instance is LC and let $\dord$ be the required linear order over $N \cup C$. Without the loss of generality, let $a \dord b$. Then we have that $v_1,v_2 \dord 3$ (otherwise, LC would be violated for voter $v_3$, a voter $j\in \{v_1, v_2\}$ such that $v_3 \dord j$, and candidates $a, b$). 

Further, suppose that $b \dord c$. Then, voters $v_1$ and $v_3$ together with candidates $b, c$ witness the violation of LC, a contradition. Hence, $c\dord b$. But then, voters $v_2$ and $v_3$ together with candidates $b, c$ witness the violation of LC. The obtained contradition completes the proof.
\end{example}

%The problem of the emptiness of the core for SSC elections remains an open problem.

\section{Finding Core-Stable Committees for Restricted Domains}\label{sec:algorithm}

In this section, we describe an algorithm for finding committees that takes as input preferences represented as weak orders. We will show that if the preferences are approval linearly consistent (LC), or strict recursive single-top-crossing (r-STC), then the returned committee belongs to the core. Our algorithm works in polynomial time. As a corollary of this result, we get that the core is always nonempty for the following classic domain restrictions: 
\begin{inparaenum}[(1)]
\item voter-interval,
\item candidate-interval,
\item single-peaked, and
\item single-crossing preferences.
\end{inparaenum}

Hereinafter we assume that the fraction $\nicefrac{n}{k}$ is integral. It does not limit the scope of the work due to the following observation:
\begin{observation}\label{obs:fractions_integral}
Consider an election $E$ and the instance $E'$ obtained from $E$ by multiplying each voter $k$ times. If a committee $W$ is not in the core for $E$, then it is not in the core for $E'$.
\end{observation}

The algorithm, which we call \textsc{CommitteeCore}, consists of two phases: first we construct a fractional committee and then we discretize it. The first phase (further called the \textsc{BestRepresentative} algorithm) is the following: imagine that each voter has an equal probability portion $\nicefrac{k}{n}$ to distribute, and that we want to choose one candidate (her \myemph{representative}) who gets this portion. Initially, the fractional committee $p$ is empty. We iterate over the set of voters, sorted according to the relation $\dord$. Let us denote by $P_i$ the set of unelected candidates at the moment of considering voter $i\in N$. The representative of $i$ is defined as a candidate $r_i\in P_i$ such that for each $c\in P_i$ it holds that either $r_i \succ_i c$ or that $r_i \sim_i c$ and $r_i \dord c$. Next, $p(r_i)$ is increased by $\nicefrac{k}{n}$. Note that, as $\nicefrac{n}{k}$ is integral, the election probability of each candidate does not exceed $1$.

In \Cref{sec:random-core} we prove that after this phase the obtained fractional committee $p$ is in the core for all strict elections and all LC approval elections. Denote by $W_1$ the set of candidates $c$ such that $p(c) = 1$. Before the second phase of the algorithm, remove candidates from $W_1$ from the election together with the voters who are represented by them, obtaining a smaller election $E_2$. By $k_2$ we denote $k-|W_1|$ (remaining seats in the committee) and by $n_2$ we denote $n-|W_1|\cdot \nicefrac{n}{k}$ (remaining voters). Renumerate the voters so that they are numbers from $[n_2]$ (and in case of r-STC elections, resort them so that $E_2$ is still r-STC). Note that by definition $\nicefrac{n_2}{k_2} = \nicefrac{(n-|W_1|\cdot \nicefrac{n}{k})}{(k-|W_1|)} = n\cdot \nicefrac{(1-\nicefrac{|W_1|}{k})}{(k-|W_1|)} = \nicefrac{n}{k}$.

The second phase (further called the \textsc{MedianRule} algorithm) is simple: for each $q\in [k_2]$ denote by $m_q$ the voter $(q-1)\cdot \nicefrac{n}{k}+1$. Further we will refer to these voters as \emph{median voters}. Then elect committee $W_2=\{r_{m_q}\colon q\in [k_2]\}$.

Finally, we return the committee $W = W_1 \cup W_2$. In \Cref{sec:discrete-core} we show that the final committee $W$ belongs to the core for LC and r-STC preferences.

\subsection{Core Stability for Fractional Committees}\label{sec:random-core}

In this subsection we prove that the committee elected by \textsc{BestRepresentative} is always in the core for LC approval elections and for all elections with strict preferences. The proof is the same for those two models; we will refer only to the following property:

\begin{definition}\label{def:property}
Given an election $E=(N, C, k)$, we say that $E$ is \myemph{well-ordered}, if there exists a linear order $\dord$ over $N \cup C$ such that for each voters $i, j \in N$ ($i \dord j$) and candidates $a, b\in C$ ($a \dord b$), if $a \sim_j b$ and $a, b\notin \bot_j$, then $a \succsim_i b$.
\end{definition}

It is clear that every strict election is well-ordered for every order $\dord$ (the premise is never satisfied). For approval elections this definition is a weakening of \Cref{def:lc} (because for approval elections $a, b\notin \bot_j \implies a\in \top_j$), hence every LC election is well-ordered.

For convenience, for $i \in N$ by $p_i$ we denote the fractional committee $p$ after considering voter $i$. Let $\sigma_i\in [m]$ be the number such that $r_i \in \pos_i(\sigma_i)$. From how the algorithm \textsc{BestRepresentative} works, we have that for every voter $i\in N$ and a candidate $c\in \pos_i([\sigma_i-1])$ it holds that $p_{i-1}(c) = 1$ (and also $p_j(c) = 1$ for every $j \geq i$).

Before proving that \textsc{BestRepresentative} returns committees belonging to the core, let us start from the following observation.

\begin{observation}\label{obs:randomized-quantified}
For each $i\in N$ and $c\in C$, there exists $q\in [\nicefrac{n}{k}]$ such that $p_i(c)=q\cdot \nicefrac{k}{n}$. In particular, $q$ is the number of voters for whom $c$ is a representative,
\end{observation}

\begin{theorem}\label{thm:randomized-Core}
Each fractional committee elected by \textsc{BestRepresentative} belongs to the core for well-ordered elections.
\end{theorem}
\begin{proof}
We will prove the following invariant: for each $i \in N$, $p_i$ satisfies the condition of the fractional core (see \Cref{def:fractional-core}) with the additional restriction that $S\subseteq [i]$. We will prove the invariant by induction.

For the first voter the invariant is clearly true. Assume, there exists $i\in N$ satisfying the invariant. We will prove that the invariant holds also for voter $(i+1)$.

For the sake of contradiction suppose that there exists a group $S \subseteq [i+1]$ and a fractional committee $p_{i+1}'$ such that for each $v \in S$ we have that $v$ prefers lexicographically $p_{i+1}'$ to $p_{i+1}$. 

First, note that if $(i+1) \notin S$, then the invariant does not hold also for $i$, a contradiction. This is the case becasue the election probability of no candidate is decreased during a loop iteration. Hence, $(i+1)\in S$. 

By the definition of \textsc{BestRepresentative} we have that for each $\varrho < \sigma_{i+1}$ and $c\in \pos_{i+1}(\varrho)$ it holds that $p_{i+1}(c) = 1$. From that, in particular we have the following equation:
\begin{equation*}
    \forall \varrho < \sigma_{i+1}.\ p_{i+1}'(\pos_{i+1}(\varrho)) \leq |\pos_{i+1}(\varrho)| = p_{i+1}(\pos_{i+1}(\varrho))
\end{equation*}
Hence, as $(i+1)$ prefers lexicographically $p'$ to $p$:
\begin{equation}\label{eq:equal_at_pos}
    \forall \varrho < \sigma_{i+1}.\ p_{i+1}'(\pos_{i+1}(\varrho)) = p_{i+1}(\pos_{i+1}(\varrho))
\end{equation}
It also needs to hold that:
\begin{equation}\label{eq:geq_at_pos}
    p_{i+1}'(\pos_{i+1}([\sigma_{i+1}])) > p_{i+1}(\pos_{i+1}([\sigma_{i+1}]))
\end{equation}

We can conclude that $\sigma_{i+1} < d_{i+1}$, as otherwise voter $(i+1)$ could not prefer  $p'_{i+1}$ over $p_{i+1}$. Consequently:
\begin{equation}\label{eq:r_not_the_lowest}
    r_{i+1} \notin \bot_{i+1}
\end{equation}

Suppose that $p_{i+1}'(r_{i+1}) = 0$. From \eqref{eq:geq_at_pos} and the fact that for all $c\in \pos_{i+1}(\sigma_{i+1})$ with $c\dord r_{i+1}$ we have $p(c)=1$, we infer that there exists $a\in \pos_{i+1}(\sigma_{i+1})$ such that $r_{i+1}\dord a$ and $p_{i+1}'(a) > 0$. From \Cref{obs:randomized-quantified} we have that $p_{i+1}'(a) \geq \nicefrac{k}{n}$. Now we modify $p_{i+1}'$ by moving the fraction of $\nicefrac{k}{n}$ from $a$ to $r_{i+1}$. By \Cref{def:property} and \eqref{eq:r_not_the_lowest} we have that for every $v\in S$ (naturally, $v\dord (i+1)$) it holds that $r_{i+1} \succsim_v a$. Thus, after the change $p_{i+1}'$ still witnesses core violation for~$S$.

Now consider a fractional committee $p_i'$ obtained from $p_{i+1}'$ by decreasing the probability portion of $r_{i+1}$ by $\nicefrac{k}{n}$. We will show that $p_i'$ together with $S\setminus \{(i+1)\}$ witness the core violation for $p_i$. Indeed, the election probability of no candidate except $r_{i+1}$ changed, and the election probability of $r_{i+1}$ changed in the same way: in $p_{i+1}$ and $p_{i+1}'$ it is higher by $\nicefrac{k}{n}$ than in $p_i$ and $p_i'$, respectively. Hence, if for a voter $v\in S$ it holds that $p_{i+1}' \vartriangleright_v p_{i+1}$, then also $p_i' \vartriangleright_v p_i$. Besides, we have that $p_i'(C) \leq k \cdot \nicefrac{|S-1|}{n}$, so we obtain a contradiction with our inductive assumption.
\end{proof}

\subsection{Core Stability for Discrete Committees}\label{sec:discrete-core}

We will now prove our main results: that the committee $W$ elected by \textsc{CommitteeCore} is in the core for approval LC preferences and for strict r-STC preferences. The algorithm for these two restricted domains is the same, but the proof techniques used for these models differ significantly.

%Here, contrary to the previous subsection, there is no reasonable definition of \emph{well-ordering}, generalizing these both classes---although the algorithm is the same, the proof techniques used in case of these models differ significantly.

\subsubsection{Core Stability for Approval LC Elections}\label{sec:algorithm-app}

Let us start with the following observation.

\begin{observation}\label{obs:order_preserved}
Consider an approval LC election $E$ and two voters $i, j$ who were not removed from the election after the first phase, such that $i \dord j$. Then either $r_i = r_j$ or $r_i \dord r_j$. 
\end{observation}
\begin{proof}
Towards a contradiction assume that $r_j \dord r_i$. From LC we have that $i$ approves $r_j$ and $r_j$ should be $i$'s representative.
% ($r_j$ and $r_j$ were not removed from the election after the first phase of the algorithm, hence $p(r_j), p(r_i) < 1$), a contradiction.
\end{proof}

Second, we prove that algorithm \textsc{CommitteeCore} elects exactly $k$ candidates.

\begin{lemma}
\textsc{CommitteeCore} for an approval LC election $E$ elects exactly $k$ candidates. 
\end{lemma}
\begin{proof}
We will show that \textsc{MedianRule} elects exactly $k_2$ candidates.
Suppose for the sake of contradiction that there are two median voters $i, j$ in $E_2$ such that $r_i = r_j$. Without loss of generality assume $i \dord j$. Consider now any voter $v$ between these median voters. If $r_v \dord r_i$ then from the definition of LC, $i$ approves $r_v$, and so $r_v$ should be selected as $i$'s representative, a contradiction. If $r_j = r_i \dord r_v$, then from the definition of LC, $v$ approves $r_j$, and so $r_j$ should be selected as $v$'s representative, a contradiction. Hence, $r_v = r_i$. But then we have that after running \textsc{BestRepresentative}, $r_i$ was a representative for at least $\nicefrac{n}{k}$ voters and was not elected, a contradiction.
\end{proof}

Note that every LC election remains LC for the same order $\dord$ after removing any number of voters and candidates. 

%Assume for convenience that we implement \textsc{MedianRule} in a sequential way: we perform $k_2$ times the operation of electing $r_i$ for some median voter $i$ and at the same time we set $p(c)=0$ for every partially elected candidate $c$ such that $r_i \dord c$.  

Finally, we are ready to prove the main technical lemma together with the main result.

\begin{lemma}\label{lem:not-much-decrease}
For each voter $i \in N$ it holds that $|W \cap \top_i| + 1 > p(\top_i)$.
\end{lemma}
\begin{proof}
Consider a voter $i\in N$. Define $\part_i$ as $p(\top_i)-|W_1 \cap \top_i|$. As $W_1$ contains all candidates $c$ such that $p(c) = 1$, then $\part_i$ is intuitively the joint sum of election probabilities of partially elected candidates in $\top_i$. From \Cref{obs:randomized-quantified} we have that:
\begin{equation}\label{eq:part_i}
    \part_i=q\cdot \nicefrac{k}{n}
\end{equation}
where $q$ is the number of voters for whom a candidate from $\top_i \setminus W_1$ is a representative. Naturally, such voters could not be removed from the election after the execution of \textsc{BestRepresentative}.

We will prove that $\part_i < |W_2\cap \top_i|+1$. From the fact that $W=W_1 \cup W_2$ and $W_1\cap W_2 = \emptyset$, it will imply the desired statement. We will now focus on upper-bounding $q$ from~\eqref{eq:part_i}.

Consider three voters $v_1, v_2, v_3$ such that $v_1 \dord v_2 \dord v_3$ and $r_{v_1}, r_{v_3} \in \top_i$. We will prove that then also $r_{v_2}\in \top_i$. Indeed, from \Cref{obs:order_preserved} we have that either $r_{v_2} \in \{r_{v_1}, r_{v_3}\}$ (and the statement is true) or $r_{v_1} \dord r_{v_2} \dord r_{v_3}$. 
First, consider the case, when $v_2 \dord i$. Since $i$ approves $r_{v_1}$ by LC applied to voters $v_2$, $i$ and candidates $r_{v_1}$ and $r_{v_2}$, we get that also $v_2$ approves $r_{v_1}$, a contradiction with \Cref{obs:order_preserved}. Second, we look at the case when $i \dord v_2$. From LC applied to $v_2$, $i$ and candidates $r_{v_2}$ and $r_{v_3}$ and by the fact that $i$ approves $r_{v_3}$ we get that $i$ also approves $r_{v_2}$, which is what we wanted to prove.

%If $i \dord v_2$, then also $r_i \dord r_{v_2}$. But analogously from LC, if $i$ approves $r_{v_3}$ then she also approves $r_{v_2}$.

Hence, these $q$ voters from \eqref{eq:part_i} need to form a consistent interval among all non-removed voters. Besides, we know that there is no more than $|W_2 \cap \top_i|$ median voters inside this interval and that between each two median voters there is $\nicefrac{n}{k}-1$ non-removed voters. Hence:
\begin{equation*}
    q \leq (|W_2 \cap \top_i| + 1) \cdot (\nicefrac{n}{k} - 1) + |W_2 \cap \top_i| = (|W_2 \cap \top_i| + 1) \cdot \nicefrac{n}{k} - 1
\end{equation*}
and:
\begin{equation*}
    \part_i = q\cdot \nicefrac{k}{n} < (|W_2 \cap \top_i| + 1) \cdot \nicefrac{n}{k} \cdot \nicefrac{k}{n} = |W_2 \cap \top_i| + 1
\end{equation*}
which completes the proof.
\end{proof}

\begin{theorem}\label{thm:discrete-Core}
For approval LC elections, \textsc{CommitteeCore} elects committees from the core.
\end{theorem}
\begin{proof}
We know that fractional committee $p$ elected by \textsc{BestRepresentative} belongs to the core.
Suppose now that $W$ is not in the core. Hence, there exists a nonempty set $S\subseteq N$ and  a committee $T$ of size $|S|\cdot \nicefrac{k}{n}$ such that $|W \cap \top_i| < |T \cap \top_i|$ for each $i \in S$---alternatively, $|W \cap \top_i| + 1 \leq |T \cap \top_i|$.

From \Cref{lem:not-much-decrease} we know that for each voter $i \in S$ we have $p(\top_i) < |W \cap \top_i| + 1 \leq |T \cap \top_i|$. Let us define a fractional committee $p'$ such that $p'(c) = 1$ for $c \in T$ and $p'(c) = 0$ otherwise. Hence, $S$ and $p'$ witness also the violation of the core for $p$, which is contradictory with \Cref{thm:randomized-Core}.
\end{proof}

\subsubsection{Core Stability for Strict r-STC Elections}\label{sec:algorithm-strict}

We will now assume that $E$ is a strict r-STC election. Similarly as in case of approval preferences, we start by proving that the \textsc{CommitteeCore} algorithm elects exactly $k$ candidates.

\begin{lemma}\label{lem:exactly-k-stc}
\textsc{CommitteeCore} for strict r-STC election $E$ elects exactly $k$ candidates
\end{lemma}
\begin{proof}
We need to show that \textsc{MedianRule} elects exactly $k_2$ candidates.
Suppose for the sake of contradiction that there are two median voters $i, j$ in $E_2$ such that $r_i = r_j$. From STC it follows that $r_v = r_i$. But this means that after running \textsc{BestRepresentative}, $r_i$ was a representative for at least $\nicefrac{n}{k}$ voters and was not elected, a contradiction.
\end{proof}

Now we prove a general statement about the application of \textsc{MedianRule} to STC elections. 
%Together with \Cref{lem:exactly-k-stc}, it will imply \Cref{cor:push-back}---that \textsc{MedianRule} returns committees from the core.

\begin{lemma}\label{lem:median-rule}
Consider an STC election $E=(N, C, k)$ and apply \textsc{MedianRule} to $E$ to obtain the committee $W$. If $|W|=k$, then $W$ is in the core.
\end{lemma}
\begin{proof}
Towards a contradiction suppose that the statement of the lemma is not true. Without loss of generality, assume that $E$ is an election with the smallest $k$ among those for which the statement of the lemma does not hold. 
Let $S$ and $T$ be subsets of voters and candidates, respectively, that witness that the committee returned by the median rule does not belong to the core. 

Observe that there are at least two candidates from $W$ that do not belong to $T$. Indeed, if there were only one such candidate, we would have that $|T|=|W|$ (as $T \setminus W$ is nonempty) and $|S|=n$. In particular, in such a case all median voters would belong to $S$. Consequently, the most preferred candidates of the median voters would belong to $T$, hence $W \subseteq T$, a contradiction.

Let us fix a candidate $a\in W \setminus T$ that is elected by the greatest median voter $(i\cdot \nicefrac{n}{k} + 1)$. In particular, $i\neq 0$. For a candidate $b \in T$ by $S_b\subseteq S$ we denote the subset of voters in $S$ preferring $b$ to $a$. Since $E$ is single-top-crossing, it holds that either $S_b \subseteq [i\cdot \nicefrac{n}{k}]$ or $S_b \subseteq N \setminus [i\cdot \nicefrac{n}{k}]$.

Now we split $E$ into two smaller elections $E_{\low}=([i\cdot \nicefrac{n}{k}], C, i)$ and $E_{\grt}=(N \setminus [i\cdot \nicefrac{n}{k}], C, k-i)$. By $W_{\low}$ and $W_{\grt}$ we denote the committees elected by the median rule for $E_{\low}$ and $E_{\grt}$, respectively. Observe that $W_{\low} \sqcup W_{\grt} = W$.

Let us also split $S$ and $T$ into two parts, as follows:
\begin{align*}
    &S_{\low} = S \cap [i\cdot \nicefrac{n}{k}], &&S_{\grt} = S \cap (N \setminus [i\cdot \nicefrac{n}{k}]), \\
    &T_{\low} = \{c\in T\colon S_c \subseteq [i\cdot \nicefrac{n}{k}]\}, &&T_{\grt} = \{c\in T\colon S_c \subseteq N\setminus [i\cdot \nicefrac{n}{k}]\} \text{.}
\end{align*}
Note that $S_{\low} \cup S_{\grt} = S$ and $T_{\low} \cup T_{\grt} = T$. Hence, if we had that both $|T_{\low}| > |S_{\low}|\cdot \nicefrac{n}{k}$ and $|T_{\grt}| > |S_{\grt}|\cdot \nicefrac{n}{k}=(|S|-|S_{\low}|)\cdot \nicefrac{n}{k}$, then we would have also $|T| > |S|\cdot \nicefrac{n}{k}$, a contradiction. Hence, for at least one of the pairs $(S_{\low}, T_{\low}), (S_{\grt}, T_{\grt})$ the opposite inequality holds. Without the loss of generality, assume that $|T_{\low}| \leq |S_{\low}|\cdot \nicefrac{n}{k}$.

We claim that the pair $(S_{\low}, T_{\low})$ witnesses the core violation for $E_{\low}$ and committee $W_{\low}$. 

Consider a voter $j \in S_{\low}$. We know that there exists a candidate $c\in T \setminus W$ such that $c\succ_j W \setminus T$.
First observe that $W_{\low}$ and $T_{\grt}$ are disjoint---indeed, for every candidate $b \in T_{\grt}$ we have that $S_b \subseteq N \setminus [i\cdot \nicefrac{n}{k}]$. As a result, there is no median voter in $[i\cdot \nicefrac{n}{k}]$ who prefers $b$ to $a$, hence $b \notin W_{\low}$. From this fact we conclude that $W_{\low} \setminus T_{\low} = W_{\low}\setminus T \subseteq W \setminus T$. Consequently, $c\succ_j W_{\low} \setminus T_{\low}$. 

Further, observe that  $c\in T_{\low}$. Indeed, voter $j$ prefers $c$ to $W \setminus T$, thus in particular $j$ prefers $c$ to $a$. Consequently, $j \in S_c$, and thus $S_c \subseteq S_{\low}$, from which we get that $c\in T_{\low}$.
Since $c\in T_{\low}$ and  $c\succ_j W_{\low} \setminus T_{\low}$, we get that $j$ prefers lexicographically $T_{\low}$ to $W_{\low}$.

Finally, we obtain that if the core was violated for $E$, it also needs to be violated for $E_{\low}$, which is contradictory to our assumption that $E$ minimizes the value of $k$.
\end{proof}

\begin{corollary}\label{cor:push-back}
In \textsc{CommitteeCore} algorithm, the committee $W_2$ is in the core for election $E_2$.
\end{corollary}

Now we are ready to prove the main theorem in this subsection:

\begin{theorem}\label{thm:core_r_stc}
Committees elected by \textsc{CommitteeCore} are in the core.
\end{theorem}
\begin{proof}
For the sake of contradiction suppose that the statement of the theorem is not true. Then there exist a set $S \subseteq N$ and a set $T \subseteq C$ witnessing the violation of the condition of the core. For every candidate $c\in C$, by $R(c)$ we denote set $\{i\in N\colon r_i = c\}$. Note that for a candidate $c \in W_1$ and a voter $i \in S$ such that $i \in R(c)$, we have $c \in T$. Hence, 
\begin{align*}
S \cap \bigcup_{c \in T \cap W_1} R(c) = S \cap \bigcup_{c \in W_1} R(c).
\end{align*}
Consider now sets $S \cap N_2$ and $T \cap C_2$ (recall that $E_2=(N_2, C_2, k_2)$ is the instance obtained after the first step of our algorithm). It holds that:
\begin{align*}
    |T \cap C_2| &= |T| -  |T \cap W_1| \leq |S|\cdot \nicefrac{k}{n} - \left|\bigcup_{c \in T \cap W_1} R(c) \right|\cdot \nicefrac{k}{n} \leq |S|\cdot \nicefrac{k}{n} - \left|S \cap \bigcup_{c \in T \cap W_1} R(c) \right|\cdot \nicefrac{k}{n} \\
                     &\leq |S \setminus \bigcup_{c \in W_1} R(c)| \cdot \nicefrac{k}{n}  = |S \cap N_2| \cdot \nicefrac{k}{n} =  |S \cap N_2| \cdot \nicefrac{k_2}{n_2} \text{.}
\end{align*}

Further, for each voter $i \in S \cap N_2$ we have that: 
\begin{align*}
T \extord_i W \implies (T\setminus W_1) \extord_i (W \setminus W_1) \implies (T \cap C_2) \extord_i W_2 \text{.}
\end{align*}
Consequently, $S \cap N_2$ and $T \cap C_2$ witness the violation of the core condition for committee $W_2$, which is contradictory to \Cref{cor:push-back}.
\end{proof}

\begin{corollary}
The core is always nonempty and can be found in polynomial time for the following classes of voters' preferences: 
\begin{inparaenum}[(1)]
\item voter-interval,
\item candidate-interval,
\item single-peaked, and
\item single-crossing preferences.
\end{inparaenum}
\end{corollary}

In \Cref{ex:tm-core-violation} below we show that the condition of recursiveness in the definition of the class of r-STC preferences is necessary for the existence of the core. Thus, in a way \Cref{thm:core_r_stc} gives a rather precise condition on the existence of the core for strict voters' preferences.  For approval preferences one cannot easily argue that the conditions are precise, since it is still a major open question whether a core-stable committee exists in each approval election.

\begin{theorem}\label{ex:tm-core-violation}
There is a top-monotonic election with strict preferences, where the core is empty. 
\end{theorem}
\begin{proof}
Let $A$ be a Condorcet cycle consisting of $r = 100$ candidates:
\begin{align*}
&a_1 \succ a_2 \succ \ldots \succ a_r \\
&a_2 \succ a_3 \succ \ldots \succ a_r \succ a_1 \\
&\ldots \\
&a_r \succ a_1 \succ a_2 \succ \ldots \succ a_{r-1}
\end{align*}
Now let $B, C, D, E$ and $F$ be five clones of $A$. Thus in $A \cup B \cup \ldots \cup F$ we have $6r = 600$ candidates. We add two more candidates, namely $g$ and $h$.

Consider the following profile with $600$ voters:
\begin{align*}
g \succ A \succ B \succ C \succ D \succ E \succ F \succ h \\
g \succ B \succ C \succ A \succ E \succ F \succ D \succ h \\
g \succ C \succ A \succ B \succ F \succ D \succ E \succ h \\
h \succ D \succ E \succ F \succ A \succ B \succ C \succ g \\
h \succ E \succ F \succ D \succ B \succ C \succ A \succ g \\
h \succ F \succ D \succ E \succ C \succ A \succ B \succ g
\end{align*}
For example, the first two votes in this profile are:
\begin{align*}
&g \succ a_1 \succ a_2 \succ \ldots \succ a_r \succ b_1 \succ b_2 \succ \ldots  \succ b_r  \succ  \ldots \ldots  \succ  f_1 \succ f_2 \succ \ldots \succ f_r  \succ h \\
&g \succ a_2 \succ a_3 \succ \ldots \succ a_1 \succ b_2 \succ b_3 \succ \ldots  \succ b_1  \succ  \ldots \ldots  \succ  f_2 \succ f_3 \succ \ldots \succ f_1  \succ h 
\end{align*}
The above profile is single-top-crossing since there are only two top-candidates, $g$ and $h$, and each of them crosses with each other candidate only once.

Let $k = 7$, and consider a committee $W$. We will show that $W$ does not belong to the core. Without loss of generality, we can assume that $g, h \in W$, as there exists more than $\nicefrac{600}{7}$ voters who rank each of these candidates as their favourite one. Further, since the profile is symmetric, without loss of generality we can also assume that it contains at most two candidates  from $A \cup B \cup C$. If the two candidates belong to the same clone, say $A$, then we take a candidate $c \in C$, and observe that 200 voters (the second and the third group) prefer $\{c, g\}$ over $W$. Otherwise, if the two candidates are from two different clones, say $A$ and $B$ (the situation is symmetric), then we take the clone which is preferred by the majority (in this context $A$) and select the candidate $a \in A$ that is preferred by $r-1$ voters to the member of $W \cap A$. There are $2r-2 = 198$ voters who prefer $\{g, a\}$ to $W$. Thus, $W$ does not belong to the core.
\end{proof}

\section{Extensions, Discussion and Open Questions}

In this work we have determined the existence of core-stable committees for a number of restricted domains both in the approval and in the ordinal models of voters' preferences. We have shown a polynomial time algorithm that returns committees belonging to the core for a subdomain of top-monotonic preferences, which includes single-peaked and single-crossing domains, and for a new class of voters' preferences which includes voter-interval and candidate-interval domains. At the same time we have shown top-monotonic ordinal elections in which no committee belongs to the core. We have additionally presented a number of results that give better insights into the structures of the known domains. In particular, our results give a better understanding of the class of top-monotonic preferences.

We conclude with one interesting observation and one important open question.

\subsection{Core and (Full) Local Stability}\label{sec:max-core}

\citet{aziz2017condorcet} proposed the concept of \emph{full local stability}, which is equivalent to the definition of core-stability for ordinal preferences. Interestingly, while the concept of the core has been studied before in the context of ordinal committee elections, the equivalence of the two concepts has never been claimed so far.
Yet, most of the results in the work of \citet{aziz2017condorcet} are formulated for the concept of \emph{local stability}. Interestingly, this concept is also equivalent to the core-stability, but for a different preference extension: we say that voter $i$ weakly prefers $W$ over $T$ according to $\extord^{\max}$ preference extension if and only if she ranks her most preferred candidate in $W$ as high as her most preferred candidate in $T$. In words, according to the $\extord^{\max}$ extension we focus only on the single top preferred candidate in the committee and do not break ties lexicographically.

%Its definition is clearly equivalent to our one of \textsc{Lex}-core. However, it may appear to be non-trivial that the second weaker definition from this paper---\emph{local stability}---corresponds to the idea of \textsc{Max}-core.

\begin{definition}[Local stability]
Consider an election $E$ and a value $q \in \mathbb{Q}$. A committee $W$ violates local stability for quota $q$ if there exists a group $S\subseteq N$ with $|S| \geq q$ and a candidate $c \in C \setminus W$ such that each voter from $S$ prefers $c$ to each member of $W$.
 %otherwise, $W$ provides local stability for quota $q$.
\end{definition}

\begin{proposition}
Local stability for quota $\lceil\nicefrac{n}{k}\rceil$ is equivalent to core stability for the $\extord^{\max}$ preference extension. 
\end{proposition}
\begin{proof}
The fact the core stability with $\extord^{\max}$ implies local stability is straightforward---local stability is a special case of the core condition for $|T|=1$. Now consider any election $E$ and a committee $W$ that is not core stability with $\extord^{\max}$.
Let $S\subseteq N$ and $T \subseteq C$ be the witness that $W$ is not in the core. For a candidate $c \in T$ let $R_c \in S$ denote a set of voters $i$ such that $c \succ_i W$. Since for every $i \in S$ there exists $c \in T$ with $i \in R_c$:
\begin{equation*}
    |S| \leq \sum_{c \in T} |R_c|
\end{equation*}

Hence, there exists a candidate $c \in T$ such that $|R_c| \geq \nicefrac{|S|}{|T|} \geq \nicefrac{n}{k}$. Yet, $R_c$ together with the candidate $c$ witness the violation of local stability, which completes the proof.
\end{proof}

%\subsection{Algorithms for Deciding whether Committees are Core Stable}
\subsection{Open Questions}

In \Cref{sec:restricted_domains} we have shown that the classic committee election rules that are commonly considered proportional are not core-stable even if the voters' preferences come from certain restricted domains. Since these domains are natural and can be intuitively explained, one would expect a good rule to behave well for such well-structured elections. On the other hand, we often require a rule which is well-defined for all preference profiles. This leads us to the following important open question.

\begin{question}\label{q:is_there_a_rule}
Is there a natural voting rule that satisfies the strongest axioms of proportionality, and which at the same time satisfies the core for restricted domains.
\end{question}

The requirement that a rule should be ``natural'' says in particular that its definition cannot conditionally depend on whether the election at hand comes from a restricted domain or not. \Cref{q:is_there_a_rule} is valid for both approval and ordinal voters' preferences.

Additionally, it would be interesting to check how often the classic rules violate the core, especially in the case  of restricted domains. One can make such a quantitative comparison via experiments. This however raises the algorithmic questions of how hard it is to verify if a given committee (in our case the committee returned by the particular rule) belongs to the core. This question is easy for the $\extord^{\max}$ preference extension.
 
\begin{proposition}
There exists a polynomial-time algorithm for deciding whether a given committee belongs to the core for the $\extord^{\max}$ preference extension.
\end{proposition} 
\begin{proof}
Given a committee $W$ it is sufficient to iterate over all candidates $c \in C \setminus W$ and check if the number of voters who prefer $c$ over $W$ is no-greater than $\lceil\nicefrac{n}{k}\rceil$.
\end{proof}

However, for the lexicographic preference extension the question is much less obvious.

\begin{question}
What is the computational complexity of deciding whether a given committee belongs to the core (assuming the standard lexicographic preference extension)? 
\end{question}

This question is interesting in the general case, and as well as for each preference domain studied in this work.

\subsection*{Acknowldegments}
Grzegorz Pierczy\'nski and Piotr Skowron were supported by Poland's National Science Center grant UMO-2019/35/B/ST6/02215.

\bibliographystyle{plainnat}
\bibliography{core_linear}

\end{document}